\documentclass[a4paper,11pt]{article}

\usepackage{authblk}
\usepackage{fullpage}
\title{The blue pebbling cost  and the space in tree-like and negative Resolution} 

\author{Lisa-Marie Jaser}
\author{Jacobo Tor\'an}

\affil{Universität Ulm, Germany\\
\{\texttt{lisa-marie.jaser},\texttt{jacobo.toran}\}\texttt{@uni-ulm.de}}

\usepackage{embrac}

\usepackage{amsmath}
\usepackage{amssymb}
\usepackage{amsthm} 	
\usepackage{amsfonts}
\usepackage{mathrsfs}

\usepackage{enumerate}

\usepackage{lmodern} 

\usepackage{microtype}  
\usepackage{url}

\usepackage{graphicx}
\usepackage{xcolor}

\usepackage[]{mdframed}
\usepackage{amsmath}
\usepackage{amssymb}
\usepackage{amsthm} 	
\usepackage{amsfonts}
\usepackage{mathrsfs} 
\usepackage{multirow}
\usepackage{tikz}
\usetikzlibrary{arrows.meta}

\usepackage{bussproofs} 
\usepackage[boxempty]{stmaryrd}
\newcommand{\emptycl}{\boxempty}

\newcommand{\MemConfig}{\mathbb{M}} 
\newcommand{\restimestep}{i} 
\newcommand{\resendtime}{t} 

\newcommand{\RefGraph}{G_{\proofstd}}

\usepackage{enumerate}

\usepackage{lmodern} 

\usepackage{microtype}  
\usepackage{url}

\newcommand{\mystyle}{\sffamily}

\usepackage{sectsty}
\allsectionsfont{\mystyle}

  \renewenvironment{abstract}{%
        \small
        \begin{center}%
          {\bfseries \mystyle \abstractname\vspace{-.5em}}%
        \end{center}%
        \quotation
      }

\date{}

\usepackage{embrac}
\usepackage{float}

\usepackage{graphicx} 

\newcommand{\proofstd}{\ensuremath{\pi}}

\newcommand{\nat}[1]{[#1]}          

\DeclareMathOperator{\pred}{\mathrm pred}

\newtheorem{theorem}{Theorem}
    \newtheorem{lemma}[theorem]{Lemma}
	\newtheorem{corollary}[theorem]{Corollary}

\theoremstyle{definition}
	\newtheorem{definition}[theorem]{Definition}
	
\DeclareMathOperator{\proj}{\mathsf {proj}}
\DeclareMathOperator{\mproj}{\mathsf {mproj}}
\DeclareMathOperator{\Mproj}{\mathsf {Mproj}}

\DeclareMathOperator{\NMproj}{\mathsf {NMproj}}

\DeclareMathOperator{\SAT}{\mathsf{SAT}}

\DeclareMathOperator{\CNF}{\mathsf{CNF}}

\DeclareMathOperator{\vars}{\mathsf{Vars}}

\DeclareMathOperator{\size}{\mathsf{Sz}}

\DeclareMathOperator{\cspace}{\mathsf{Cs}}

\DeclareMathOperator{\pebspace}{\mathsf{peb-space}}

\DeclareMathOperator{\BW}{\mathsf{BW}}
\DeclareMathOperator{\Pebbling}{\mathsf{P}}

\DeclareMathOperator{\Black}{\mathsf{Black}}
\DeclareMathOperator{\RMc}{\mathsf{RB}}

\DeclareMathOperator{\Rev}{\mathsf{Rev}}
\DeclareMathOperator{\blue}{\mathsf{Blue}}
\DeclareMathOperator{\PD}{\mathsf{PD}}
\DeclareMathOperator{\Blue}{\mathsf{Blue}}

\DeclareMathOperator{\Peb}{\mathsf{Peb}}

\DeclareMathOperator{\RES}{\mathsf{Res}}

\newcommand{\derivabbrevsmall}[2]{{( #1 \vdash #2 )}} 
\newcommand{\derivabbrevcompact}[2]{{\bigl( #1 \vdash #2 \bigr)}} 

\newcommand{\refutabbrevsmall}[1]{\derivabbrevsmall{#1}{\!\emptycl}}
\newcommand{\refutabbrevcompact}[1]{\derivabbrevcompact{#1}{\!\emptycl}}

\newcommand{\derivof}[4][\derives]
{{\ensuremath{{#2} : {#3} \, {#1}\, {#4}}}}
\newcommand{\refof}[2]{\derivof{#1}{#2}{\emptycl}} 

\newcommand{\genericformsmall}[2]{\mathrm{#1}( #2 )}

\newcommand{\genericrefsmall}[3]%
{{\mathrm{#1}}_{#2}\refutabbrevsmall{#3}}
\newcommand{\genericrefcompact}[3]%
{{\mathrm{#1}}_{#2}\refutabbrevcompact{#3}}
\newcommand{\genericderiv}[4]%
{{\mathrm{#1}}_{#2}\derivabbrev{#3}{#4}}
\newcommand{\genericderivsmall}[4]%
{{\mathrm{#1}}_{#2}\derivabbrevsmall{#3}{#4}}
\newcommand{\genericderivcompact}[4]%
{{\mathrm{#1}}_{#2}\derivabbrevcompact{#3}{#4}}
\newcommand{\generictaut}[3]%
{{\mathrm{#1}}_{#2}\derivabbrev{}{#3}}
\newcommand{\generictautcompact}[3]%
{{\mathrm{#1}}_{#2}\derivabbrevcompact{}{#3}}
\newcommand{\generictautsmall}[3]%
{{\mathrm{#1}}_{#2}\derivabbrevsmall{}{#3}}
\newcommand{\arbitraryrefsmall}[3]%
{{\mathscr{#1}}_{#2}\refutabbrevsmall{#3}}

\newcommand{\clspaceabbr}{CS}
\newcommand{\clspaceof}[2][]{\genericformsmall{\clspaceabbr_{#1}}{#2}}
\newcommand{\clspaceref}[2][]{\genericrefsmall{\clspaceabbr}{#1}{#2}}

\newcommand{\treeclspaceabbr}{Tree\text{-}\clspaceabbr}
\newcommand{\treeclspaceref}[2][]{\genericrefsmall{\treeclspaceabbr}{#1}{#2}}

\newcommand{\setsofvarsorlit}[2]%
{\mathrm{#1}({#2})}

\newcommand{\Nclspaceabbr}{N\text{-}\clspaceabbr}
\newcommand{\Nclspaceref}[2][]{\genericrefsmall{\Nclspaceabbr}{#1}{#2}}

\newcommand{\Pclspaceabbr}{P\text{-}\clspaceabbr}
\newcommand{\Pclspaceref}[2][]{\genericrefsmall{\Pclspaceabbr}{#1}{#2}}

\newcommand{\gamevalueformat}[1]{\mathsf{#1}}

\newcommand{\PebConfig}{\mathbb{P}}

\newcommand{\White}{\gamevalueformat{White}}

\newcommand{\pebtimestep}{i} 
\newcommand{\pebendtime}{t} 

\begin{document}

\maketitle

\begin{abstract}
The red--blue pebble game 
is a well known two-player game on graphs that has been used in the past as a tool to analyze complexity measures in several computation models as well as proof systems. We define a new way to measure the cost of the game, the blue cost, which only counts  the number of pebbles that are colored blue during the game.  This new measure characterizes exactly several space bounds in tree-like and negative Resolution.
 
In particular we prove that for 
any unsatiafiable formula $F$,  the clause space requirements of the formula in  tree-like Resolution, exactly coincide with the minimum blue pebbling cost of the game played on a refutation graph of $F$ (not necessarily a tree). This exactly parallels the known result for general Resolution in terms of the standard black pebble game, and improves the existing approximated characterization of tree-like space in terms of reversible pebbling.

We show that the blue pebbling cost is also well suited for analyzing the space requirements of the lifted pebbling formulas
$\Peb_G[\vee]$ and $\Peb_G[\oplus]$ in the two Resolution restrictions. In the case of tree-like Resolution, 
the clause space of $\Peb_G[\vee]$ asymptotically coincides  with the blue cost of the underlying graph $G$. 
For the case of negative Resolution, we obtain almost matching upper and lower bounds for the space in the two classes of lifted formulas, similar to the ones existing for general Resolution. 

We also prove a close to optimal space separation between tree-like and negative Resolution, presenting a class of formulas with $n$ variables that require clause space $\Omega(\frac{n}{\log n})$ in negative Resolution, but have constant space tree-like refutations. This contrasts with the fact that negative Resolution can simulate tree-like Resolution with only a small increase in size.
\end{abstract}

\section{Introduction}
Proof complexity  studies the resources, such as proof size or space, required  for proving or refuting mathematical statements in concrete proof systems. As a bridge between logic and complexity theory, proof complexity uses tools from both 
fields, and 
like in other areas of logic or complexity, the  use of games for the analysis of resources has here also a long tradition. 
A class of games that has been very useful for the analysis of computational models is that of pebble games
\cite{PH70Comparative,Cook74ObservationTimeStorageTradeOff,Sethi75CompleteRegisterAllocation,HPV77TimeVsSpace}.  
This is   particularly true for analyzing complexity measures in proof systems.
Intuitively, the idea of the pebble games is to measure the minimum number of pebbles (the game price) needed by a single player to place a pebble on the single sink of a directed acyclic graph  following certain rules. 
Inspired by the pebble game, a class of contradictory formulas, called pebbling formulas, was introduced in~\cite{BW01ShortProofs}. These formulas have been extremely useful for analyzing several proof systems, as can be seen, for example, in
\cite{Nordstrom06NarrowProofsMayBeSpaciousSTOC,BN11UnderstandingSpaceFULLREF, Nordstrom12RelativeStrength, Nordstrom13SurveyLMCS, ToranW21, RezendeMNR21,JaserT26}.
The reason for this is that some of the pebbling properties of the underlying graphs are transformed into parameters for the complexity of their corresponding pebbling contradictions. Known results in pebbling can  be translated  in this way into proof complexity upper and lower bounds. 
The original pebbling formulas are not hard to refute. 
For the mentioned results, the basic  formulas have been lifted in several ways, substituting each variable in the formula by a combination of new variables, thus increasing its complexity. The most used substitutions are the OR and the $\oplus$ substitutions, in which a variable $x$ is substituted by the subformula $x_1\vee x_2$, respectively  $x_1\oplus x_2$, for two new variables $x_1,x_2$.

In the black pebble game on a single sink directed acyclic graph  (sDAG) $G$, pebbles can only be placed on a vertex if this is a source or if all its direct predecessors already have a pebble, but these pebbles can be removed at any time. In the reversible pebble game, it is also required that a pebble can only be removed when all its predecessors have a pebble on them. We denote the  costs of these two games played on $G$  
by $\Black(G)$ and $\Rev(G)$ and clearly for every $G$, $\Black(G)\leq \Rev(G).$
These are single-player games, but it was proved in \cite{Chan13JustAPebble}, that somehow unexpectedly the reversible pebble game is equivalent (in terms of game price) to a two-player pebble game called the red--blue game. There are different equivalent versions of this game \cite{DT85,Chan13JustAPebble}, the  one used here, also known as the Raz--McKenzie game,  was introduced in \cite{RM99Separation}.  In the red--blue pebble game, two players, Pebbler and Colorer decide the colors of some of the vertices of a sDAG $G$.  
 In the first round, Pebbler places a pebble on the sink and Colorer colors the pebble red. In all subsequent rounds, Pebbler places a pebble on an arbitrary empty vertex of $G$ and Colorer colors this new pebble either red or blue. The game ends when there is a vertex with a red pebble that is either a source vertex or all its direct predecessors in the graph have blue pebbles.
The price of the game on $G$,  $\RMc(G)$ is the smallest number~$r$ such that Pebbler has a strategy to make the game end in at most $r$ rounds regardless of how Colorer plays. Intuitively we can think or strategies for Pebbler thinking that he has to pay one point in every round, and he wants the game to finish as soon as possible.

We introduce in this paper a new price or measure in the red--blue game, the blue price of $G$ or $\Blue(G)$.  The game does not change, just the way to measure its price. This is just the number of pebbles that are colored blue before the game finishes.
We can imagine that Pebbler only has to pay for the blue pebbles and the game does not have to finish quickly, as long as not too many pebbles are colored blue. As we will see, this measure lies between the black and reversible pebbling costs, for every sDAG $G$, $\Black(G)\leq \Blue(G)\leq \Rev(G).$ We will also show graph families for which the first and second measures, or the second and third measures are asymptotically different.  Our motivation for introducing the new pebbling prize is its suitability as a tool for analyzing  the clause space complexity measure in some restrictions of the Resolution proof system.

Resolution is due to its simplicity and its many applications in the analysis of $\SAT$-solvers, the best studied proof system, and  
several complexity measures for the analysis of Resolution refutations like the size, the space or the clause width have been  defined in the last decades. 
In this paper, we will  concentrate on the clause space, \clspaceabbr,  defined as the number of clauses required simultaneously in a refutation. This measure gives an estimate of the minimum memory requirement needed in the refutation of a formula.
A different interpretation of the clause space in a Resolution refutation is the  cost of the black pebble game played on the refutation graph \cite{ET01SpaceBounds}. A  connection between pebbling and clause space that differs  conceptually from this one is the
fact that most of the known space bounds in Resolution have been proven for  the pebbling formulas or for its lifted versions. 


Many restrictions of the general Resolution  system have been introduced in order to restrict the search space for proofs, maybe at the cost of increasing the proof size (see e.g. \cite{KBL99}). Tree-like Resolution and  negative Resolution are two of the best known restrictions. In the first one,
the underlying refutation graph has to be a tree. This restriction has been intensively studied since it is essentially equivalent to the DPLL $\SAT$-solving algorithm.   
In the case of negative Resolution, in every Resolution step, one of the resolved clauses must contain only negative literals.  In terms of proof efficiency, negative Resolution lies between general and tree-like Resolution \cite{BEGJ00RelativeComplexity,BP07Complexity}. Both restrictions are sound and refutation complete.

\subsection{Related results}

It is known that general Resolution can be much more efficient than tree-like or negative Resolution in terms of size (number of clauses in a refutation) \cite{BIW04Near-optimalSeparation,BP07Complexity}. In the last reference it is also shown that  there are  formulas with  negative refutations of polynomial size, but only exponential size 
tree-like Resolution refutations \cite{BEGJ00RelativeComplexity,BP07Complexity}. However such a separation in the other direction does not exist, and negative Resolution  polynomially simulates tree-like Resolution. In fact it is not hard to see that for every unsatisfiable formula $F$ with $n$ variables and  with a tree-like resolution of size $S$, there is also a negative refutation of size  $nS$. 

Space separations between general Resolution and its restrictions and are much more modest. In~\cite{JMNZ12RelatingProofCplx,ToranW21}, 
the authors give examples of families of formulas with $O(n)$ variables that can be solved in clause space  $s$ (for several functions $s$) but require $\Theta(s\log n)$ tree-like Resolution space. We do not know of any result regarding  the  clause space of  negative Resolution refutations.

As mentioned before, space bounds in general Resolution for the lifted pebbling  formulas have been well studied. 
With considerable effort,  space lower bounds for the $\Peb_G[\vee]$ formulas, for the special cases of trees and pyramid graphs, matching the black pebbling number of the underlying graph,  were given in  
\cite{NH08TowardsOptimalSeparationSTOC}  (see also \cite{Nordstrom08}). 
The techniques were very much simplified in \cite{BN11UnderstandingSpaceFULLREF} with the use of $\Peb_G[\oplus]$ formulas.
The authors proved  that the black-white\footnote{Since our results do not use  this version of the game, we have not defined it here and refer the interested reader to the references.} pebbling cost lower bounds the space requirements for refuting these formulas, for every
DAG $G$. Later in \cite{Nordstrom12RelativeStrength} it was shown that this lower bound is optimal. The question of  whether this is also an upper bound,  as well as  the existence of a lower bound for the $\Peb_G[\vee]$ formulas that works for general DAGs, 
are interesting  open problems. 

For the case of tree-like Resolution it was shown in \cite{ToranW21} that the reversible pebbling cost of the underlying graph
$G$ exactly characterizes the space required for refuting $\Peb_G[\oplus]$.  For this result, yet another game, the Prover-Delayer game played on formulas was used. This game was introduced in \cite{PudlakI00} as a tool for proving size lower bounds in tree-like Resolution, and shown to exactly characterize tree-like space in \cite{ET03CombinatorialCharacterization}.

\subsection{Our results}

We use the concept of blue pebbling cost in order  to analyze and characterize the clause space complexity measure in tree-like and negative Resolution. 

In Section~\ref{sec:blue-tree-like} we prove that for any unsatisfiable
$F$, the tree-like Resolution space of $F$ exactly coincides with
the minimum blue pebbling cost of a refutation graph for $F$,
$$\treeclspaceref{F}  =\min_{\refof{\proofstd}{F}} \blue(\RefGraph).$$ 
Interestingly, the minimum is taken over all possible refutation graphs  of~$F$, not only over the trees. This improves a result in  \cite{ToranW21} where  the reversible pebbling number $\Rev(G_\pi)$ 
was used to 
approximate the clause space in tree-like Resolution, by a logarithmic factor. The result completely parallels the characterization 
of clause space in \cite{ET01SpaceBounds}, with tree-like instead of general Resolution and with the blue pebbling cost instead of the black cost.

The rest of the results are related to  the lifted pebbling formulas $\Peb_G[\vee]$ and $\Peb_G[\oplus]$.
In Section~\ref{sec:blue-v-tree-like} we show that the blue pebbling cost also  characterizes the space needed in the 
tree-like refutation of  $\Peb_G[\vee]$. For this we show that $\Blue(G)$ is a lower bound for the cost of the Prover--Delayer game played on $\Peb_G[\vee]$. This result improves the lower bound of $\Omega(\Black(G))$  for the cost of the Prover--Delayer game on $\Peb_G[\vee]$ given in \cite{BIW04Near-optimalSeparation}.

Section~\ref{sec:negative-res} is devoted to the space bounds for the refutation of the lifted pebbling formulas 
in N-Resolution. As in the case of  the existing results for general Resolution, we are not able to obtain matching upper and lower bounds for the clause space, but the difference between the bounds is not large. The most involved result is the space 
lower bound for the 
$\Peb_G[\vee]$ formulas. This follows the lines of the space lower bound for  $\Peb_G[\oplus]$ in general resolution
given in \cite{BN11UnderstandingSpaceFULLREF}, adapting the concepts in the proof to the context of N-Resolution, but with the difficulty that unlike $\oplus$, $\vee$ is a non-authoritarian function. 
We first show that the clause space  in a (general) refutation of $\Peb_G[\vee]$ is at least as large as the number of negated literals in something called the negative projection of the formulas, we then show that
from the negative projection it is possible to obtain an N-Resolution refutation of the non-lifted formula $\Peb_G$ with the
same negative width, and finally show that $\Black(G)$ is a lower bound for the negative width of  $\Peb_G$.
The space lower bound is a consequence of these three results.
For the upper bound, the  concept of the  blue pebbling number plays also a central role. 

The case of the space lower bound for the $\oplus$ lifted formulas is simpler, and uses  special partial assignment of the variables to flip the polarity of the variables in $\Peb_G$. We show that the space requirements in N-Resolution for the formulas obtained this way, are at least as large as the black pebbling price of the underlying graph.
A summary of the space upper and lower bound results for the refutation of  the lifted pebbling formulas can be seen in Table 1.


 \begin{tiny}
 \begin{table}[H]
 \centering
\begin{tabular}{ |c|c|c|c|}
\hline
& \textbf{Resolution}&\textbf{tree-like $\RES$}&\textbf{negative $\RES$}\\
\hline
\multirow{2}{*}{$\Peb_G[\vee]$}                 & {$O(\Black(G))$} &  \multirow{2}{*}{$\Theta(\Blue(G))^\dagger$}                   & {$O(\Blue(G))^\dagger$}                \\ 
  & Open &  &$\Omega(\Black(G))^\dagger$ \\
\hline
\multirow{2}{*}{$\Peb_G[\oplus]$}         & {$O(\Black(G))$}          & \multirow{2}{*}{$\Theta(\Rev(G))$}                & {$O(\Rev(G))^\dagger$}      \\ 
& $\Omega(\BW(G))$ & & $\Omega(\Black(G))^\dagger$\\
\hline
                                                                              
\end{tabular}
\bigskip
\caption{The known upper and lower bounds for the lifted pebbling formulas. The results marked with $\dagger$ have been obtained in this work.}
\end{table}
\end{tiny}

The results in the table show that for the case of the pebbling formulas,  the space requirements in tree-like or negative Resolution are not very different. We investigate in Section~\ref{sec:comparing} the space separations between the two Resolution restrictions for general formulas.
As mentioned,  the largest known difference between general and tree-like Resolution space is by a logarithmic factor. We also achieve this factor showing a family of formulas that can be refuted with constant space in N-Resolution, but require logarithmic 
space in T-Resolution.

In the other direction, we show that there are formulas with $n$ variables and  size  $O(n)$  with constant space refutations in T-Resolution, but which require space $\Omega(\frac{n}{\log n})$ in N-Resolution. This separation is almost best possible since, as
in general Resulution, 
$n+2$ is a clause space upper bound for a formula with $n$ variables in N-Resolution.
This separation  is remarkable because, as mentioned above, it is known
that in terms of size N-Resolution can simulate T-Resolution at the expense of a small size increase.

We conclude with some open problems. 

\section{Preliminaries}

\subsection{Pebble Games}
Black pebbling was first mentioned implicitly in~\cite{PH70Comparative}.
Note that there exist several variants of the (black) pebble game in the literature.
For a detailed description of these variants, we refer to~\cite{Nordstrom09PebblingSurvey}. 
For the following definitions, let $G = (V,E)$ be a DAG with a unique sink vertex $z$.
\begin{definition}[Black pebble game]
    The \emph{black pebble game} on $G$ is the following one-player game:
At any time $\pebtimestep$ of the game, there is  a \emph{pebble configuration} $\PebConfig_\pebtimestep := B_\pebtimestep$,
where $B_\pebtimestep \subseteq V$ is the set of black pebbles.
A black pebble configuration $\PebConfig_{\pebtimestep-1} = B_{\pebtimestep-1}$ can be changed to $\PebConfig_\pebtimestep = B_\pebtimestep$ by applying exactly one of the following rules:
\begin{description}
    \item[Black pebble placement on $v$:] If all direct predecessors of an empty vertex $v$ have pebbles on them, a black pebble may be placed on $v$. 
	\item[Black pebble removal from $v$:] A black pebble may be removed from any vertex at any time. 
\end{description}
A \emph{black pebbling} of $G$ is a sequence of pebble configurations $\Pebbling = (\PebConfig_0,\PebConfig_1,\dots,\PebConfig_\pebendtime)$ such that $\PebConfig_0 = \PebConfig_\pebendtime = 
\emptyset,$ for some $i\leq \pebendtime,$ $z\in B_i$, and for all $\pebtimestep \in [\pebendtime]$ it holds that $\PebConfig_\pebtimestep$ can be obtained from $\PebConfig_{\pebtimestep-1}$ by applying one of the above-stated rules.
Observe that w.l.o.g. we can always assume that $B_{\pebendtime-1}=\{z\}$. \\
For convenience, we will also use the dual notion of  \emph{white pebble} game.
A {white pebbling} is a pebbling where $\PebConfig_\pebtimestep := W_\pebtimestep$ with $W_\pebtimestep \subseteq V$ is a set of white pebbles for all $\pebtimestep \in [\pebendtime]$. 
A \emph{white pebbling} is defined similarly, but the placement and removal rules are switched.\\
Notice that  $\Pebbling = (\PebConfig_0,\PebConfig_1,\dots,\PebConfig_\pebendtime)$ is a black pebbling of $G$ if and only if 
$\Pebbling' = (\PebConfig'_t,\dots,\PebConfig'_0)$ is a white
pebbling of $G$, where  each configuration $\PebConfig'_i$ contains the same 
set of pebbled vertices as in $\PebConfig_i$, but with white pebbles instead 
of black pebbles. In a white pebbling we can always suppose that $W_1=\{z\}.$
\end{definition}

\begin{definition}[Pebbling space, and price]
	The \emph{space} of a pebbling 
    $\Pebbling = (\PebConfig_0,\PebConfig_1,\dots,\PebConfig_\pebendtime)$  
     $\pebspace(\Pebbling) := \max_{\pebtimestep \in \nat{\pebendtime}} | \PebConfig_\pebtimestep |$.
	The \emph{black pebbling price} of~$G$, denoted by $\Black(G)$, is the minimum space of any black pebbling of~$G$. By the observation above, 
 the white pebbling price $\White(G)$ coincides with $\Black(G)$
\end{definition}
We also consider  the reversible pebble game introduced in~\cite{Bennett89TimeSpaceReversible}. In the reversible pebble game, the moves performed in reverse order should also constitute a legal black pebbling, which means that the rules for pebble placements and removals are symmetric. A pebble can only be removed if all its predecessors have pebbles on them.
This implies that reversible pebbling is a restricted
version of black pebbling.
The notion of 
reversible pebbling price, $\Rev(G)$ is defined as in the other pebbling variants.

The red--blue game is played on DAGs and was introduced in \cite{RM99Separation} as a tool for studying the depth complexity of decision trees for search problems. Contrary to the previous pebble games, this one is a two-player game.

\begin{definition}[Red--blue game]
    \label{def:RazMcKenzieGame}
    The \emph{red--blue game} is played on a sDAG $G$ by two players, \emph{Pebbler} and \emph{Colorer}. The game is played in rounds, where Pebbler and Colorer alternate. In the first round, Pebbler places a pebble on the sink and Colorer colors the pebble red. In all subsequent rounds, Pebbler places a pebble on an arbitrary empty vertex of $G$ and Colorer colors this new pebble either red or blue. The game ends when there is a red pebbled vertex  that is either a source vertex or all its direct predecessors in the graph have blue pebbles.
\end{definition}

\begin{definition}[Red--blue price]
    \label{def:RazMcKenziePrice}
    The \emph{red--blue price} $\RMc(G)$ of a sDAG~$G$ is the smallest number~$r$ such that Pebbler has a strategy to make the game end in at most~$r$~rounds regardless of how Colorer plays.
\end{definition}

In~\cite{Chan13JustAPebble} it was shown that the reversible pebbling price and the red--blue price coincide for any single-sink DAG.

\begin{theorem}[\cite{Chan13JustAPebble}]
    For any single-sink DAG $G$ it holds $\RMc(G) = \Rev(G)$. 
\end{theorem}

Finally, we define a game introduced in \cite{PudlakI00} for the analysis of resources in tree-like Resolution. This game is played on formulas.

\begin{definition}[Prover--Delayer game]
	\label{def:ProverDelayerGame}
	The \emph{Prover}-\emph{Delayer} \emph{game} 
	between two players, called \emph{Prover}, and \emph{Delayer}, is
	played on an unsatisfiable CNF formula~$F$ in rounds.
	Each round starts with Prover querying the value of a variable.
	Delayer can give one of three answers:~$0$,~$1$,~or~$\ast$.
	If~$0$ or~$1$ is chosen by Delayer, no points are scored by her and the queried variable is set to the chosen value. 
	If Delayer answers~$\ast$, then Prover gets to decide the value of that variable, and Delayer scores one point.
	The game finishes when any axiom in $F$ has been falsified (all its literals are set to~$0$) by the partial assignment constructed this way.
	If this is not the case, the next round begins. The aim of Delayer is to win as many points as possible, while Prover aims to minimize{} this quantity.
\end{definition}

\begin{definition}[Game value of the Prover--Delayer game]
	Let $F$ be an unsatisfiable CNF formula.
	The game value of the Prover--Delayer game played on~$F$, denoted by $\PD(F)$, is the greatest number of points Delayer can score on~$F$ against an optimal strategy of Prover.
\end{definition}

The Prover--Delayer game  exactly characterizes{} the tree-like clause space of a formula. 

\begin{theorem}[\cite{ET03CombinatorialCharacterization}]
	\label{thm:TreeClauseSpaceProverDelayerCharacterization}
	For any unsatisfiable CNF formula $F$,
$\treeclspaceref{F} = \PD(F)+2.$
\end{theorem}

\subsection{The blue pebbling price}

We consider a new way to define the  price in the red--blue pebble game. This is  just the number of pebbles in the game that are colored blue by Colorer (+1). 

\begin{definition}[Blue pebble price]\label{def:blue-cost}
    The \emph{blue pebbling number} of a single sink DAG $G$,  $\Blue(G)$ is the smallest number $b+1$ such that Pebbler has a strategy
    to make the red--blue game on $G$ end,  with at most $b$ blue pebbles, regardless of Colorer's strategy. The extra 1 is added for technical reasons in order to make the results cleaner. For example, if $G$ has just one vertex we will say $\Blue(G)=1$ and if
    $G$ consists of one edge, then $\Blue(G)=2$.
    \end{definition}

Since in the blue pebbling price the cost of the red moves in the $\RMc$ game is not counted, the blue cost cannot be larger than the reversible pebbling number. We will show that this number cannot be smaller than the black pebbling number and therefore
$\Black(G)\leq \Blue(G)\leq \Rev(G).$
We will provide examples of graph families $\{G_n\}$ and $\{H_n\}$
for which there is an asymptotic difference between the values of $\Rev(G_n)$ and $\blue(G_n)$,
as well as between   $\Black(H_n)$ and $\blue(H_n)$.

\subsection{Formulas}

We consider only propositional formulas in conjunctive normal form, CNF. A literal is a variable or its negation. 
For a set of variables $\{x_1,\dots,x_n\}$, a clause is a (non-tautological) disjunction of literals, and it is represented
by a set of literals. A formula is a conjunction of clauses, and it is represented by the set of its clauses. We denote by $\vars(F)$
the set of variables in $F$.

For a partial assignment  $\alpha:\vars(F)\rightarrow \{0,1\}$, and a clause $C\in F$,
$C\!\restriction\!_\alpha$ denotes the result of substituting in $C$ the variables assigned in $\alpha$ by its value. If some literal in $C$ is assigned value 1 in $\alpha$, then $C\!\restriction\!_\alpha=1$ otherwise $C\!\restriction\!_\alpha$ is the clause
after deleting all literals assigned value 0 by $\alpha$. For a set of clauses $F$, $F\!\restriction\!_\alpha$
denotes the result of applying $\alpha$ to all the clauses in $F$.

\begin{definition}[Pebbling formulas]
Let $G=(V,E)$ be a DAG with a set of sources $S \subseteq V$ and a unique sink $z$. We identify every vertex $v \in V$ with the Boolean variable $v$ (and the lifted variables $v_1$ and $v_2$). For a vertex $v\in V$ we denote by  $\pred(v)$ the set of its direct predecessors. In particular, for a source vertex $v$, $\pred(v)=\emptyset$. The \emph{pebbling contradiction} over $G$, denoted $\Peb(G)$, is the conjunction of the following clauses:
	\begin{itemize}
		\item for all  vertices $v$, the clause 
  $\bigvee_{u \in \pred(v)} \bar u \lor v$, 
         \hfill (\emph{pebbling axioms})
		\item 
  for the unique sink $z$, the unit clause $\bar z$. \hfill (\emph{sink axiom})
	\end{itemize}
\end{definition}
A well known method to make  formulas  harder to refute is to substitute some suitable Boolean function~$f(x_1,\dots,x_{d})$ of arity $d$ for each variable $x$ and expand the result into $\CNF$ ($x_1,\dots,x_{d}$ are new variables). We say that the function $f$ is \emph{non-authoritarian} if no partial assignment of any subset  of less than $d$ variables can fix the value of $f$ to 0 or 1, and $f$ is authoritarian otherwise. 
We restrict ourselves to the arity $d=2$ and the special cases of $f(x_1,x_2)=x_1\vee x_2$
and $f(x_1,x_2)=x_1\oplus x_2.$ Observe that  $\oplus $ is non-authoritarian, while $\vee$ is an authoritarian function.

\begin{definition}[Substitution formulas]
    For a $\CNF$ formula $F$  and a function $f\in\{\vee,\oplus\}$ we define $F[f]$ to be the formula resulting from 
    substituting every variable $x$ in $F$ by $f(x_1,x_2)$ and writing the result as a $\CNF$.
    \end{definition}
    
    For example, for the clause $C=(x\vee\bar y)$, $C[\vee]$ is the conjunction of the clauses $(x_1\vee x_2\vee \bar y_1)$ and 
    $(x_1\vee x_2\vee \bar y_2)$, while $C[\oplus]$ is defined by the four clauses
    $ (x_1 \lor x_2\lor   {y_0} \lor \bar{y_2}), (x_1 \lor x_2\lor  \bar y_1 \lor {y_2}),(\bar x_1 \lor \bar  x_2\lor   {y_1} \lor \bar{y_2}),$ and 
    $(\bar x_1 \lor \bar x_2\lor\bar    {y_1} \lor {y_2})$.

\subsection{Resolution basics}
\begin{definition}[Resolution]
    Let $F$ be a CNF formula. A \emph{Resolution derivation} of a clause $C$ in $F$ (write $F \vdash C$) is defined as a sequence of clauses $\pi= (C_1,...,C_m)$ such that $C_m = C$ and for all $i=1,...,m$, $C_i$ is either 
    \begin{itemize}
        \item a clause in the formula $F$ (an axiom),
        \item $C_i\supseteq C_j$ with $j<i$  (a weakening of a previous clause), or
        \item is obtained by applying the  Resolution rule to some $C_j$ and $C_k$, $j,k< i$:
        \begin{equation*}
            \frac{A\lor x \qquad B\lor \overline{x}}{A\lor B}.
        \end{equation*}
        In this last case $A\vee B$ is called the resolvent of the previous two clauses.
    \end{itemize}
    When $C$ is the empty clause $\square$, the proof is called \emph{refutation} of $F$ (write $F\vdash\Box$).
\end{definition}

The weakening rule is not strictly necessary but simplifies some proofs. 
A Resolution derivation $\pi$ of a clause $C$ can be represented as an inference DAG $G_\pi$ in which each clause is identified with a vertex in $G_\pi$, and there is an edge from a clause
$A$ to another $B$, whenever $B$ is a resolvent of  $A$. The vertex corresponding to the axioms do not have any predecessors,
and the derived  clause $C$  is a sink in $G_\pi$.

For the study of space in Resolution, we consider configurational proofs. 

\begin{definition}[Configuration-style Resolution]
\label{def:ConfigurationStyleResolution}
A \emph{configurational derivation } of a clause $C$ from a CNF formula~$F$ is an ordered sequence of \emph{memory configurations} (sets of clauses)~$\proofstd = (\MemConfig_0,\dots,\MemConfig_{\resendtime})$ such that $\MemConfig_0 = \varnothing$, $C \in \MemConfig_{\resendtime}$ and for each $\restimestep \in \nat{\resendtime}$, the configuration~$\MemConfig_{\restimestep}$ is obtained from $\MemConfig_{\restimestep-1}$ by applying exactly one of the following rules:
\begin{description}
\item[Axiom Download:] $\MemConfig_\restimestep = \MemConfig_{\restimestep-1} \cup \{D\}$ for some axiom clause $D \in F$.
\item[Erasure:] $\MemConfig_\restimestep = \MemConfig_{\restimestep-1} \setminus \{D\}$ for some $D \in \MemConfig_{\restimestep-1}$.
\item[Inference:] $\MemConfig_\restimestep = \MemConfig_{\restimestep-1} \cup \{D\}$ for some resolvent $D$ inferred from $C_1, C_2 \in \MemConfig_\restimestep$ by the resolution rule.
\end{description}
\end{definition}


\begin{definition}[Complexity measures for Resolution]
\label{def:ClauseSpace}
The \emph{size} of a Resolution derivation $\proofstd = (C_1,\dots,C_t)$,  $\size(\pi)=t$ is the number of clauses used in the derivation.

The \emph{clause space} of a memory configuration $\MemConfig$, $\clspaceof{\MemConfig} = | \MemConfig|$, is  
the number of clauses in $\MemConfig$. 
The clause space of a configurational refutation $\proofstd = (\MemConfig_0,\dots,\MemConfig_{\resendtime})$ is defined by $\clspaceof{\proofstd} := \max_{\restimestep \in \nat{\resendtime}} \clspaceof{\MemConfig_\restimestep}$.

Taking the minimum over all refutations of a formula $F$, we define $\size({F\vdash }) := \min_{\refof{\proofstd}{F}} \size({\proofstd})$, and  $\clspaceref{F}:= \min_{\refof{\proofstd}{F}} \clspaceof{\proofstd}$ 
 as the length and  clause space of refuting $F$ in Resolution.
\end{definition}

In this work, we consider the following restrictions of  general Resolution.

\begin{definition}[Tree-like Resolution]
    A \emph{tree-like Resolution} derivation of a clause $C$ from a formula $F$ is a Resolution derivation on $C$ in which the inference  graph is a tree. This implies  that a clause has to be re-derived if it is used more than once in a proof.
\end{definition}
\begin{definition}[Negative and  positive Resolution]
    A \emph{negative Resolution} derivation of a clause $C$ from a  formula $F$ is a Resolution derivation on $C$ in which in every 
    inference  step one of the  parent clauses  consists only of negative literals.\\
    In  the dual concept of \emph{positive Resolution}  in each inference step  a parent clause  must  have only positive literals.
\end{definition}

The size and space complexity measures are extended to the Resolution restrictions in the natural way. For an unsatisfiable formula $F$, we denote clause space of  $F$ in tree-like and negative Resolution by 
 $\treeclspaceref[]{F}$ and $\Nclspaceref[]{F}.$  For the analysis of proofs in negative Resolution we also use the concept of 
 negative width.

 \begin{definition}[Negative width]\label{def:negative-width}
   The negative width of a clause $C$ is the number of negated literals it contains. The negative width of a refutation $\pi$ is the maximum of the 
   negative width of its clauses, and the negative width of an unsatisfiable formula $F$ is the minimum of the negative width of an N-Resolution of $F$.
\end{definition}



\section{The blue pebbling cost}\label{sec:blue-number}

For a sDAG $G=(V,E)$,  a set $B\subseteq V$ and $t\in V\setminus B$,  $G_{B,t}$ denotes 
a version of the graph $G$ in which the vertices in $B$ have been colored blue, and  $t$ is  colored red. We will also assume that all vertices in $V$ that do not reach $t$ are deleted from $G$.
$\Blue(G_{B,t})$ is the blue pebbling number when the game is played on the partially colored graph $G_{B,t}$, and 
the vertices in $B$ do not count as points 
that Pebbler has to pay in the game on $G_{B,t}$. If $t$ is a source node in $G$, $t\in B$ or all predecessors of $t$ are in $B$, then $\Blue(G_{B,t})=1$ (see Definition~\ref{def:blue-cost}).

\begin{lemma}\label{lem:blue-pebbling}
\begin{enumerate}
\item For $B\subseteq V$, $t\in V$ for every $v\in V\setminus (B\cup \{t\}),$  $$\Blue(G_{B,t})\leq \max\{\Blue(G_{B,v}),  
    \Blue(G_{B\cup \{v\},t})+1\}, \hbox{ and }$$ 
    \item there is some $v\in V\setminus (B\cup \{t\}),$ such that  $\Blue(G_{B\cup\{v\},t})< \Blue(G_{B,t}).$  
\end{enumerate}
\end{lemma}

\begin{proof}
   For the first part,  the optimal strategy of Pebbler cannot be worse than  choosing $v$ as the  first vertex to pebble in 
   $\Blue(G_{B,t})$.
If Colorer answers blue, then Pebbles pays one point and the game continues on $\Blue(G_{B\cup\{v\},t})$ (this is the second term).
If Colorer answers red, no point is paid,  the game continues with the same set of blue vertices, and with the two red vertices $v,t$. The price of reaching a contradiction in this game cannot be larger than the price of the game on the graph with the unique red vertex $v$. 

For the second part, we can define a decision tree $T$ corresponding to 
an optimal strategy for Pebbler  
for the game on $\Blue(G_{B,t}).$ 
The vertices in the tree are labeled with vertices in $G$ and from each internal vertex in $T$ come out two edges,
corresponding to the two possible answers red or blue, of Colorer for $v$. Every path in the tree corresponds to a sequence of 
answers of Colorer and the leaves correspond to end positions in the game. Let $b=\Blue(G_{B,t})$. Every path from the root to a leaf in $T$ have at most $b-1$ edges labeled blue.  
The vertex at the root of the decision tree is such a vertex $v$ satisfying the inequality, since if $v$ is colored blue, we  move 
to a subtree in $T$ with from which the maximum number of blue edges in a path to a leaf is $b-2$. 
\end{proof}

\begin{lemma}\label{black<blue}
For any sDAG $G=(V,E)$, $\Black(G)\leq \Blue(G)$.
\end{lemma}
\begin{proof}
Let $G=(V,E)$ be a sDAG with $n$ vertices and sink $t$. Assume $\blue(G)=b$. We prove by induction over $n$ that $\Black(G)\leq b$.
If $n=1$, then by definition $\blue(G)=1$.

For $v\in V$ let us denote by $G_v$ the subgraph of $G$ in which $v$ is the unique sink, and all the vertices that do not reach $v$ are deleted, and let $G-v$ be the subgraph resulting after deleting $v$ and all the incident edges.
Similarly as in Lemma~\ref{lem:blue-pebbling}, for the case of black pebbling, it holds 
$\Black(G)\leq \max\{\Black(G_v),\Black(G-v)+1\}$. This is because a possible strategy for pebbling $G$ is to 
place a pebble in $v$ (for this $\Black(G_v)$ pebbles are needed), remove all the pebbles except the one in $v$, and keeping this pebble until the end, pebble $G$. For the second part of the strategy, at most $\Black(G-v)$ (plus the pebble on $v$) are needed.
By the induction hypothesis this means $\Black(G)\leq \max\{\blue(G_v),\Blue(G-v)+1\}.$

By Lemma~\ref{lem:blue-pebbling}
there must be a vertex $v\in V$ such that $\Blue(G_{\{v\},t})=\Blue(G_{\emptyset,t})-1.$
Observe also that $\Blue(G-{v})= \Blue(G_{\{v\},t})$. This is because a configuration $(B,R)$ is an ending configuration in the first graph if and only if $(B\cup\{v\},R)$ is an ending configuration in the second graph. Therefore, there exists a vertex $v\in V$ with $\Blue(G-v)<\blue(G)$ which implies
$\Black(G)\leq \max\{\blue(G_v),\Blue(G)\},$
and because $G_v$ is a subgraph of $G$, it follows  $\blue(G_v)\leq\blue(G)$.
\end{proof}

There are DAG families $\{G_n\}$ with an asymptotic difference between the values of $\Rev(G_n)$ and $\blue(G_n)$, for example, if $G_n$ is a directed path with $n$ vertices, then $\Rev(G_n)=\Theta(\log n)$ \cite{Bennett89TimeSpaceReversible}, while $\blue(G_n)=2$. Also, there are families $\{H_n\}$ of graphs in which the prices of $\Black(H_n)$ and $\blue(H_n)$ are different. An example of such a family are the road graphs of width at least 2 \cite{ChanLNV23}. These graphs are  fan-in 2 sDAGs that can be considered as generalizations of directed paths. 
\begin{definition}[Road Graph]
    The \emph{road graph} $R(s,l)=(V_R,E_R)$ of length $l$ and width $s$ ($l\geq s$) is defined as follows:
    \begin{align*}
        V_R=&\{v_{i,j}| i=1,...,s, \text{ }j=1,...,l-s+i\},\\
        E_R=&\{(v_{i,j},v_{i,j+1})|i=1,...,s, \text{ }j=1,...,l-s+i-1\}\cup\\
        &\{(v_{i,j},v_{i+1,j+1}|i=1,...,s-1, \text{ }j=1,...,l-s+i-1\}\cup\\
        &\{(v_{s,j},v_{1,j+1})|j=1,...,l-s\}
    \end{align*}
    We call the vertices $v_{i,j}$ for $j=1,...,l-s+1$ the $i$-th \emph{layer}.
\end{definition}

\begin{figure}[h]
\centering
\begin{tikzpicture}
    \fill (0,0) circle (0.1);
    \fill (0,-1) circle (0.1);
    \fill (1,0) circle (0.1);
    \fill (1,-1) circle (0.1);
    \fill (2,0) circle (0.1);
    \fill (2,-1) circle (0.1);
    \fill (3,0) circle (0.1);
    \fill (3,-1) circle (0.1);
    \fill (4,0) circle (0.1);
    \fill (4,-1) circle (0.1);
    \fill (5,-0.5) circle (0.1);
    \draw[-Stealth] (0,0) -- (0.9,0);
    \draw[-Stealth] (0,-1) -- (0.9,-1);
    \draw[-Stealth] (1,0) -- (1.9,0);
    \draw[-Stealth] (1,-1) -- (1.9,-1);
    \draw[-Stealth] (2,0) -- (2.9,0);
    \draw[-Stealth] (2,-1) -- (2.9,-1);
    \draw[-Stealth] (3,0) -- (3.9,0);
    \draw[-Stealth] (3,-1) -- (3.9,-1);
    \draw[-Stealth] (4,0) -- (4.9,-0.5);
    \draw[-Stealth] (4,-1) -- (4.9,-0.5);
    \draw[-Stealth] (0,0) -- (0.9,-1);
    \draw[-Stealth] (0,-1) -- (0.9,0);
    \draw[-Stealth] (1,0) -- (1.9,-1);
    \draw[-Stealth] (1,-1) -- (1.9,0);
    \draw[-Stealth] (2,0) -- (2.9,-1);
    \draw[-Stealth] (2,-1) -- (2.9,0);
    \draw[-Stealth] (3,0) -- (3.9,-1);
    \draw[-Stealth] (3,-1) -- (3.9,0);

\end{tikzpicture}
\caption{A road graph of length $6$ and width $2$}\label{Road graph}

\end{figure}
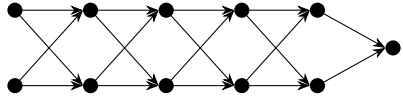

It is easy to see, that the black pebbling number of  a road graphs of width $s$ is at most $s+2$. 
For such  pebbling strategy one can places pebbles on the sources of  the graph and move towards the sink. 

It was shown in \cite{ChanLNV23} Lemma~4.5 that for any $l,s$ with $l\geq s$, $\Rev(R(s,l)=\Theta(s\cdot m)$ where
$m=\max\{\log(l/s),1\}$. We adapt their proof to show that for $s\geq 2$, this also holds for the blue pebbling number. For this we need the following definition and lemma from the mentioned reference. 

\begin{definition}[Blocking set]
    Let $G$ be a sDAG with sink $s$. A \emph{blocking set} of $G$ is a subset of vertices $B\subseteq V(G)$ such that every path from a source vertex in $G$ to $s$ contains a vertex in $B$. A blocking set $B$ is called \emph{minimal} if no subset of $B$ is a blocking set of $G$.
\end{definition}

\begin{lemma}[\cite{ChanLNV23}]\label{lem:chanLNV}
A minimal blocking set of a road graph that spreads over $d$ layers has size at least $d+q-1$, where $q$ is the width of the topmost layer. 
\end{lemma}

\begin{lemma}\label{lem:road-graph}
    For the road graph $R=R(s,l)$ with length $l$ and width $s\geq 2$ it holds that $\blue(R)=\Theta(s\cdot m)$, where $m=\max\{\log(l/s),1\}$.
\end{lemma}
\begin{proof}

    The upper bound follows from Lemma~\ref{lem:chanLNV} and the fact that for every $G$, $\Blue(G)\leq \Rev(G)$.

    
    For the lower bound we give an inductive proof over the parameter  $k=l-s$ in the road graph. The base case $l=s$ is the pyramid at the end of $R$ and the results follows from Lemma~\ref{black<blue} and the  well known fact that the black pebbling number of a pyramid of base $s=\Omega(s)$. 
    For the induction step,   assume that for some constant $c$, a road graph of width $s$ and length $l$ with $l-s=i$ has blue pebbling number at least $cs \log( l/s)$.  We give a strategy for Colorer. The game is divided in phases. In the first phase
    Pebbler and Colorer play until the pebbles chosen by  of Pebbler complete a blocking set $B$ of $R$. 
     Colorer  chooses the blue color every time, except maybe for the last vertex $v$, for which Colorer can either choose blue or red. If she chooses blue, the next phase continues 
     on the graph in which the  vertices on the lowest layer that intersects with $B$ act as the new sources
     (we suppose that the vertices that might not be colored in this layer are colored blue for free). 
     
     If Colorer chooses red, for the last vertex in the block, we consider a minimal path from the sources of $R$ to $v$. Let $u$ be the predecessor $v$  on this path. In the next phase of the game, vertex $u$  acts as the new sink and the game is played on the graph induced by the vertices in $R$ that can reach $u$.
     In both cases, we consider that all pebbles that are not in the blocking set are deleted (this  would only be a disadvantage for Colorer). We now need to give the criteria for Colorer to  decide whether 
     vertex $v$ is colored blue or red. Assume $B$ spreads over $d$ layers.\\
    The first case we need to consider is when the last layer of the $B$ has width $q<s$, i.e. the last layer is in the pyramid at the end of the graph. In this case Colorer chooses red. The road graph that remains has length $l-d-q-1$ because the number of layers below the one with width $q$ is $q-1$. The number of blue pebbles needed is 
    \[|B|+\Blue(R(s,l-q-d-1))\geq d+q-1+s\log((l-q-d-1)/s)\geq s\log(l/s).\]
    For the other two cases let $a$ be the lowest and $b$ be the highest layer of $B$ and let $m=(a+b)/2$. It holds $d=b-a+1$. If $m> l/2$, Colorer chooses red again and the remaining road graph has length $a-1$. The blue pebbling number now is
    \[|B|+\Blue(R(s,a-1))\geq d+s-1+s\log((a-1)/s)\geq d+s-1+s\log((l-d-1)/2s)\geq s\log(l).\]
    For the last case assume $m\leq l/2$. The remaining road graph has length $l-b+1$ and the blue pebbling number is
    \[|B|+\Blue(R(s,l-b+1))\geq d+s-1+s\log((l-b+1)/s)\geq d+s-1+s\log((l-d-1)/2s)\geq s\log(l).\]
    \end{proof}

From this result, and the the fact that the black price for a road graph $R(l.s)$ is at most $s+2$ we obtain the following result similar to Theorem 4.2 in \cite{ChanLNV23}. 

\begin{corollary}
   For any function $s(n)=O(n^{1/2-\epsilon})$ for a constant $0<\epsilon<\frac{1}{2}$,  there are sDAG families $\{G_n\}$of size $O(n)$ such that $\Black(G_n)=O(s(n))$ while $\Blue(G_n)=\Omega(s(n)\log n).$
\end{corollary}

\section{Blue pebbling exactly characterizes clause space tree-like Resolution}\label{sec:blue-tree-like}

In \cite{ToranW21} it was shown that the cost of the reversible pebble game played on the refutation graph of a formula, approximates the clause space in tree-like Resolution, by a logarithmic factor. We obtain an exact characterization of this complexity measure using the blue pebbling number. 
Note that the minimum in Theorem~\ref{thm:TreeCS-leq-min-RMc-RefGraph} is taken over all possible refutation graphs  of~$F$, not only over the trees. We need a intermediate result stating that the black and blue costs coincide for binary trees. 


\begin{lemma}
\label{lem:bin-tree}
	For any binary tree $T$ it holds
    $\Black(T)=\blue(T)$.
\end{lemma}
\begin{proof}
By Lemma~\ref{black<blue} we know that $\Black(T)\leq \blue(T)$.  For the other direction, let $T$ be a binary tree with $\Black(T)=k$. 
	For any node $v$ in $T$ let		
	$T_v$ be the subtree  rooted at~$v$.
	We show by induction on $d$, the depth of $T_v$, that for any vertex~$v$ in~$T,$ if $\Black(T_v)=k$		
	then there is  a strategy for		
	Pebbler in the red--blue game on $T_v$ in which at most $k$ vertices are colored blue.		
	The case $d=0$ is trivial.
	For $d>0$, the game starts, according to the rules, by Pebbler 		
	querying the root~$v$ of the subtree and Colorer{} answering~red. 		
	We consider two cases, depending on whether for both direct predecessors 
	$v_1$ and~$v_2$ of~$v$ in~$T,$ $\Black(T_{v_1})=\Black(T_{v_2})=k-1$ or not.
	In the first case,		
	Pebbler queries one of them, say~$v_1$. If the answer is red, he continues on $T_{v_1}$. By induction hypothesis $\blue(T_{v_1})\leq k-1$. Otherwise he continues on~$T_{v_2}$. By the induction hypothesis, $\blue(T_{v_2})\leq k-1$ and then $\blue(T_{v})\leq 1+k-1=k$.
    
    In case for (w.l.o.g.) $\Black(T_{v_1})=k$ and $\Black(T_{v_2})\leq k-1$, Pebbler queries $v_1$ first. If the answer from Colorer is red, he continues on $T_{v_1}$ and needs at most $k$ blue pebbles. If the answer is blue, he queries $v_2$. If the answer from Colorer is blue again, the game finishes and Pebbler has paid 2 points. Otherwise Pebbler continues on $T_{v_2}$ and uses at most $1+k-1=k$ blue pebbles.      
\end{proof}

\begin{theorem}
	\label{thm:TreeCS-leq-min-RMc-RefGraph}
	For any unsatisfiable $\CNF$ formula $F$ with $n$ variables it holds
	$$	\treeclspaceref{F}  =\min_{\refof{\proofstd}{F}} \blue(\RefGraph).$$ 
	
\end{theorem}

\begin{proof}
	For proving  $\treeclspaceref{F}  =\min_{\refof{\proofstd}{F}} \blue(\RefGraph)$ we follow
	the intuition behind the red--blue game and 
	identify 
	the color{} blue with 1 and the color{} red with 0.
    Let~$\proofstd$ be a Resolution refutation of~$F$ with a refutation-graph~$\RefGraph$ and~$\blue(\RefGraph)=k$.
	By Theorem~\ref{thm:TreeClauseSpaceProverDelayerCharacterization}, 
    it suffices to give a strategy for Prover in the Prover--Delayer game played on~$F$
	under which he has to pay at most $k-2$ points. Consider an optimal strategy $S$ for Pebbler in
   in the red--blue game played on~$\RefGraph$. In this strategy the number of pebbles colored blue
   is at most $k-1$ (the other point is counted by definition). Also, w.l.o.g. we can suppose that the game ends with a vertex in $\RefGraph$ being colored blue by Colorer (if the game is finishing no matter what color she gives, she might as well color the vertex blue and score one more point). 
  Prover just has to simulate Pebbler strategy in $S$  and we will see that the number of times 
    Delayer scores a point (by answering $*$), is  bounded by $k-2$.
    By doing so, a partial assignment~$\alpha$ falsifying an axiom in~$F$ will be produced.  
	The game is divided in stages.
	Initially the partial assignment is the empty assignment.
	In each stage, 
	if Pebbler chooses a clause $C \in \RefGraph$, Prover queries the variables in $C$ not yet assigned by $\alpha$, one by one in any order, extending the partial assignment~$\alpha$ with the answers of Delayer, until either:
	\begin{enumerate}
		\item \label{itm:CIsSatOrFalsByTva} the clause~$C$ is satisfied or falsified by~$\alpha$, or
		\item \label{itm:AVarInClauseIsGivenAstValue} a variable~$x$ in~$C$ is given value~$\ast$ by Delayer.
	\end{enumerate}
	In case~\ref{itm:CIsSatOrFalsByTva}, Prover moves to the next stage, simulating the strategy of Pebbler assuming Colorer{} has given clause~$C$ the
	color~${C}\restriction_{\alpha}$. 
	In case~\ref{itm:AVarInClauseIsGivenAstValue}, Prover extends~$\alpha$ by assigning~$x$ with the value that satisfies~$C$ and moves to the next stage, simulating the strategy of Pebbler, assuming Colorer{} has given clause~$C$ the color~$1$ (blue). Observe that the number of times a vertex is colored blue in $G_\pi$
 cannot be larger than $\blue(G_\pi)-1$ since  the colors are  being chosen according to the strategy of Delayer, and this cannot be better than the best strategy of Colorer.
 The red--blue game finishes placing  at most $k-1$ blue pebbles.  This means that  Delayer can score at~most~$k-1$~points
    by answering $*$. We will show that in fact she can score at most $k-2$ points. 
	The red--blue game finishes when either a source 
	in~$\RefGraph$ is assigned color~$0$ by Colorer, or a vertex with all its direct predecessors being colored~$1$ is colored~$0$.
    The second situation is not possible since for any partial assignment~$\alpha$ it cannot be that~$\alpha$ satisfies
	two parent clauses in a Resolution proof, while falsifying their
	resolvent.
    The first situation corresponds to~$\alpha$ falsifying an axiom $A$ in~$F$.
    Moreover, this clause can not be falsified after Delayer gives a $*$ to $A$ because in this case the produced assignment 
    $\alpha$ would satisfy $A$, so it must be the case that $A$ is falsified by assigning  a literal 
    from  another queried clause, or in other words, in the red--blue game Pebbler has not queried clause $A$, thus saving the last blue point.

	To prove that the other inequality, let $k=\treeclspaceref{F}$. 
	There is a tree-like refutation $\pi$ of $F$ whose underlying graph $G_\pi$ is a binary tree with black pebbling		price~$k$. The result follows from Lemma \ref{lem:bin-tree}.
    \end{proof}

\section{The tree-like space of $\Peb_G[\vee]$}\label{sec:blue-v-tree-like}

We give another application of the blue pebbling cost 
showing that in tree-like Resolution, the refutation space for the $\Peb_G[\vee]$ formulas asymptotically  coincides  with the blue cost of $G$.
It is interesting to compare this results with the corresponding one for the formulas lifted with the $\oplus$ gadget. In \cite{ToranW21}
it was shown that $\treeclspaceref{\Peb_G[\oplus]}=\Theta(\Rev(G)).$
\begin{theorem}\label{thm:tree-cs-peb}
Let $G$ be a sDAG. 
$\treeclspaceref{\Peb_G[\vee]}=\Theta(\blue(G)).$
\end{theorem}
\begin{proof}
    Let $G$ be a fixed sDAG. We prove that
    $\RMc(G) \leq \PD(\Peb_{G}[\vee])$ and $\PD(\Peb{G}[\vee]) \leq 2 \cdot \RMc(G)$.
    The results then follow from Theorem~\ref{thm:TreeClauseSpaceProverDelayerCharacterization}. 
    
        Let $b=\Blue(G)$. We first show the inequality $\PD(\Peb_{G}[\vee]) \leq  2b$ by giving a strategy for Prover, in which the Delayer can score at most $2b$ points. Prover  simulates the strategy of Pebbler in the red--blue game: If Pebbler pebbles a vertex~$v$ of $G$, Prover will query the variables~$v_1$ and~$v_2$ of $\Peb_{G}[\vee]$ in this order. The red--blue game ends after at most~$b$~vertices have been colored blue. We argue, that the Prover--Delayer game also ends after at most~$2b$~queries. Thus, Delayer only gets a chance to score $2b$ points. 
        In case the second variable of a pair gets queried, the best choice Delayer has is to follow the strategy of Colorer and to ensure that $v_1 \vee v_2$ is true under her constructed assignment, if~$v$ is colored~$1$, and false if~$v$ is colored~$0$. At the end of the red--blue game either a source vertex $s$ in $G$ is colored~$0$, or a vertex~$v$ of~$G$ is colored~$0$, while all its direct predecessors are colored~$1$. In the first case, the source $s$ being colored $0$ leads to the falsification of the corresponding source axiom $s[\vee]$ by Delayer. In the second case, Delayer will falsify a clause of the corresponding pebbling axioms $\big( \bigwedge_{u \in \pred_G(v)} \overline{u} \lor v \big) [\vee]$.
        
        Next, we show the inequality $\PD(\Peb_{G}[\vee]) \geq  b$ by giving a strategy for Delayer such that under any strategy of Prover, she scores at least~$b$~points. By Definition~\ref{def:RazMcKenziePrice}, there is a strategy of Colorer, in which Pebbler has to pebble $b$  vertices colored blue by Colorer, in order  to end the game. Delayer  essentially copies this strategy: The first time a variable in the  pair corresponding to a vertex $v$  gets queried, if Colorer would answer blue,  
        she can answer~$\ast$, otherwise she answers 0. The second time, she can copy the response of Colorer on $v$. 
        With this strategy she  scores at least~$b$~points. \qedhere
\end{proof}

\section{The space of lifted pebbling formulas in negative Resolution}\label{sec:negative-res}

In this section, we study the space in negative Resolution for  the pebbling formulas
lifted with the $\vee$ and with the $\oplus$ functions.  We show that for $\Peb_G[\vee]$ formulas, the space in negative Resolution lies between $\Black(G)$ and $\Blue(G)$, while for $\Peb_G[\oplus]$ formulas, the space in negative Resolution lies between $\Black(G)$ and $\Rev(G)$. It is interesting to compare these results with the existing space bounds for general Resolution shown in Table 1.

\subsection{Upper bounds}

\begin{theorem}\label{thm:upper-peb-nres-space}
Let $G$ be a sDAG. 
\begin{enumerate}
    \item 
$\Nclspaceref{\Peb_G[\vee]}\leq \blue(G)+4$
\item 
$\Nclspaceref{\Peb_G[\oplus]}\leq 2\Rev(G)+4$
\end{enumerate}

\end{theorem}

Let $t$ be the unique sink in $G$. We consider first the pebbling formulas $\Peb_G[\vee]$ and for a sets $B\subseteq V$ 
and $t\in V\setminus \{t\}$
let $A_{B,t}$ be the set of partial assignments $\alpha$ defined on the variables $t_1, t_2$  and for each $v\in B$ on exactly one of the variables in the pair
$\{v_1, v_2\}$ and such that $\alpha(t_1)=\alpha(t_2)=0 $ and $\alpha(v_i)=1.$

Let $G=(V,E)$ with sink $t$ and $B\subseteq V$, $t\in V\setminus B$ and let $\alpha\in A_{B,t}$, 
we will consider a particular kind of negative Resolution refutations of $\Peb_G[\vee]\ \restriction_\alpha$ that we call {\em rooted refutation}. In such a refutation, the first two steps download the clauses $\bar t_1, \bar t_2$ in memory,
and these two clauses stay in memory until the end of the refutation. Observe that the clause space in a rooted refutation is at most two memory units larger than an unrestricted refutation.

The first part of Theorem~\ref{thm:upper-peb-nres-space} follows from the next lemma by letting $B=\emptyset$.

\begin{lemma}
Let $G$ be a graph with single sink $t$ and let $B\subseteq V$, $t\in V\setminus B$ and let $\alpha\in A_{B,t}$, then  
there is a rooted refutation for $\Peb_G[\vee]\ \restriction_\alpha$ with clause space at most  $\Blue (G_{B,t})+4$
\end{lemma}

\begin{proof}
    By induction on $n'=|V\setminus B|$. For $n'=0$, the predecessors of $t$ in $G$ are in $B$. The clauses $t_1\vee t_2, \bar t_1$ and $\bar t_2$ belong to the formula  $\Peb_G[\vee]\ \restriction_\alpha$. There is a rooted refutation of these three clauses with clause space 4.   
    On the other hand $\Blue (G_{B,t})=1$.

For the induction step, let $B\subseteq V$ and let $k=\Blue (G_{B,t}).$
By Lemma~\ref{lem:blue-pebbling} there is some vertex $v\in V\setminus(B\cup \{t\}$ with $\Blue (G_{B\cup\{v\},t})<\Blue (G_{B,t})$.
For an assignment $\alpha\in A_{B\cup\{v\},t}$ by induction hypothesis,  
there is a rooted refutation of $\Peb_G[\vee]\ \restriction_\alpha$ with  clause space at most $\Blue (G_{B\cup\{v\},t})+4\leq\Blue (G_{B,R})+3=k+3$  (the last inequality follows by the second part of Lemma~\ref{lem:blue-pebbling}). 

This means that for $\alpha\in A_{B,t}$ it is possible to derive the clause $\bar v_1$ in negative Resolution from 
$\Peb_G[\vee]\ \restriction_\alpha$ in clause space $k+3$, and it is possible to derive the two clauses $\bar v_1$ and $\bar v_2$ in clause space $k+4$ (one after the other).

 Let $G'$ the subgraph of $G$ with single sink $v$ (the vertices not reaching $v$ in $G$ are deleted in $G'$). Also  by induction hypothesis for $\alpha\in A_{B,t}$ the formula $\Peb_G'[\vee]\ \restriction_\alpha$
can be refuted with a rooted refutation in negative Resolution and  clause space  
 at most $\Blue (G'_{B,v})+4\leq \Blue (G_{B,t})+4=k+4.$

For $\alpha\in A_{B,t}$, a strategy to refute $\Peb_G[\vee]\ \restriction_\alpha$ in negative Resolution is the following: derive 
$\bar v_0$ and $\bar v_1$ in clause space $b+4$. Delete everything from memory except these two clauses and simulate the rooted refutation
of $\Peb_G'[\vee]\ \restriction_\alpha$ but with the axioms from $\Peb_G'[\vee]\ \restriction_\alpha$.
The only two axioms from $\Peb_G'[\vee]\ \restriction_\alpha$ that are not in $\Peb_G[\vee]\ \restriction_\alpha$
are $\bar v_0$ and $\bar v_1$,
and these have been derived and kept in memory.
\end{proof}

We prove now the second part of the theorem in a similar way, but using the red-blue pebble game.

\begin{lemma}
Let $G$ be a graph with single sink $t$ and sources $S\subseteq V\setminus \{t\}$, 
and let $v\in V\setminus S$. Let $\alpha$ be an assignment of the variables $v_1,v_2$. 
\begin{enumerate}
    \item If $\alpha(v_1)+\alpha(v_2)$ is even, then $\Peb_G[\oplus]\restriction_\alpha\leq 2\RMc(G_{S,v})+4$ and 
    \item if $\alpha(v_1)+\alpha(v_2)=1$ is odd, then $\Peb_G[\oplus]\restriction_\alpha\leq 2\RMc(G_{S\cup\{v\},t})+4$
\end{enumerate}
\end{lemma}

\begin{proof}
    By induction on $n'=|V\setminus S|$. For 
For $n'=1$ the only vertex that is not a source is $t$. 
For the two predecessor $u,w$ of $t$ the clauses belong there are 8 clauses with  variables in  $\{u_1,u_2,v_1,v_2\}$ that are axioms in $\Peb_G[\vee]\ \restriction_\alpha$ and  there is negative refutation of these  clauses with clause space 4   
    On the other hand $\RMc (G_{S,t})=0$.

For the induction step, let  $k=\RMc (G_{S,t}).$
There is some vertex $v\in V\setminus(S\cup \{t\}$ with $\RMc (G_{S\cup\{v\},t})<\RMc (G_{S,t})$
and $\RMc (G_{S,v})<\RMc (G_{S,t})$. 

For the assignment $\alpha$ with $\alpha(v_1)=\alpha(v_2)=1$ by the induction hypothesis
there is negative refutation of $\Peb_G[\vee]\ \restriction_\alpha$ with  clause space at most $2\RMc (G_{S,v})+4\leq 2(\RMc (G_{S,t})-1)+4=2k+2$. 

This means that  it is possible to derive the clause $(\bar v_1\vee \bar v_2)$ in negative Resolution from 
$\Peb_G[\vee]\ \restriction_\alpha$ in clause space $2k+2$.

Analogously  for the assignments $\beta$ and $\gamma$ with $\beta(v_1)=0, \beta (v_2)=1$ and $\gamma(v_1)=1, \beta (v_2)=0$ 
there are negative refutation $R_1$ of $\Peb_G[\vee]\ \restriction_\beta$ and $R_2$ of $\Peb_G[\vee]\ \restriction_\gamma$ with  clause space at most $2\RMc (G_{S\cup\{v\},t})+4\leq 2k+2$ and also for the assignment $\delta(v_1)=\delta(v_2)=0$
there is a negative refutation $R_3$ of $\Peb_G[\oplus]\restriction_\delta$ using at most  space  $2k+2$.

A strategy to refute $\Peb_G[\oplus]$  in negative Resolution is the following: derive 
$(\bar v_1\vee\bar v_2)$ in clause space $2k+2$. Delete everything from memory except this clause 
and simulate first the refutation $R_1$ obtaining $\bar v_2$ in space $2k+3$. Keep  $(\bar v_1\vee\bar v_2)$ as well as
$\bar v_2$ in memory and simulate $R_2$ to derive $\bar v_1$ in space $2k+4$. Delete $(\bar v_1\vee\bar v_2)$ keeping 
the two clauses $\bar v_1$ and $\bar v_2$ in memory and simulate $R_3$ to refute
$\Peb_G[\oplus]$ in space $2k+4$.
\end{proof}

\subsection{Lower bounds}

We give  in this section space lower bounds for N-Resolution refutation  the lifted pebbling formulas, showing that for both kinds
of lifted formulas,
the black pebbling number of the underlying graph is a lower bound for the space. These results  are considerably more involved than the upper bounds, and are not completely tight, since as mentioned, there are known  graph families for which  the corresponding 
black and blue or reversible pebbling numbers differ by a logarithmic factor.

For the case of the $\Peb_G[\vee]$ formulas, we show that the space needed in a negative refutation must be at least as large
as the maximum number of negative literals that appear together in a clause in a  negative refutation of $\Peb_G$ (the negative width of the refutation). We then use the fact that the negative width of a negative refutation of $\Peb_G$ is at least the 
black pebbling number of $G$.
The structure of the proof as well as the next  definitions are clearly inspired  in \cite{BN11UnderstandingSpaceFULLREF}. 

 \begin{definition}\label{def:projection}\index{projection}
   Let $V$ be a set of variables, 
   and let $\mathbb D$ be a set of clauses over $V[\vee]$. The  projection of $\mathbb D$,  $\proj(\mathbb D)$ is the set of clauses 
    $C$
    with $\vars(C)\subseteq V$, 
    with the property that $C[\vee]$ is a logical consequence of $\mathbb D$, 
   $\mathbb D\models C[\vee].$
\end{definition}

Observe that if a clause $C$ is in the  projection of a $\mathbb D$ then
every weakening of $C$,  is also
in $\proj(\mathbb D)$. Since we want the set of variables in the projection
to remain small, we need to consider only the minimal projected clauses. These are
defined as follows:

\begin{definition}\label{def:min-projection}
Let $\mathbb D$ be a set of clauses over $V[\vee]$ and $C$ be a clause with $\vars(C)\subseteq V$. We say that  $C$ is a {\em minimal  clause projected} by $\mathbb D$ 
and write 
$C\in\mproj(G)$,
if $C\in\proj(G)$ but for any proper sub-clause $C'\subsetneq C,$ $C'\not\in \proj(G)$.

For a set of clauses $\mathbb D$ with $\vars(\mathbb D)\subseteq  V[\vee]$,  
we denote by $\Mproj(\mathbb D)$ the set of clauses
$C$ that are minimal projections of some subset of
clauses in $\mathbb D$
$$\Mproj(\mathbb D)=\{C\ | \ \exists\ \mathbb D'\subseteq \mathbb D, C\in \mproj(\mathbb D')\}.$$
\end{definition}


 Notice that the empty clause can be in
$\proj(\mathbb D)$ only if $\mathbb D$ is unsatisfiable.
Also, if $\mathbb D'\subseteq \mathbb D$, then
$\Mproj(\mathbb D')\subseteq \Mproj(\mathbb D)$.

  Let $G$ be a sDAG, and consider the formula $\Peb_G[\vee].$
For a set of clauses $\mathbb D$ with  variables in $\vars(F[\vee]$ we define teh negative projection $\NMproj(\mathbb D)$ as the set of clauses in 
$\Mproj(\mathbb D)$ that do not have any positive literals or are axioms in $\Peb_G$,
$$ \NMproj(\mathbb D)=\{C\in \Mproj(\mathbb D)\ |\ C \hbox{ does not have positive literals or } C\in
\Peb_G\}.$$

  The next result shows that for a  set of clauses $\mathbb D$, the number of clauses in $\mathbb D$ is always greater than the number of distinct variables that appear negated  in some clause in $\NMproj(\mathbb D)$. We will later apply this result considering that $\mathbb D$ is a configuration in an N-Resolution proof of the lifted formula $\Peb_G[\vee]$.

\begin{lemma}\label{lem:cs>nproj}
Let $\mathbb D\not=\emptyset$ be a set of clauses
with  variables in $\vars(\Peb_G[\vee])$. The number
of clauses in $\mathbb D$ is larger than
the set of negated literals in the
minimal  negative projections of $\mathbb D$
$$|\mathbb D| > |\{ x\ |\ \bar x\in C\  \mbox{\rm for some } C\in \NMproj(\mathbb D)\}|.$$
\end{lemma}

\begin{proof}
Let $V^-\subseteq V$ be the set of variables  that appear negated some  clause in $\NMproj(\mathbb D)$, 
and suppose $V^-\not=\emptyset.$ 
We construct a bipartite incidence graph with two kinds of vertices.
On the left side the vertices represent the set of clauses in 
$\mathbb D$ and on the 
right side the vertices are the variables in $V^-.$ There is an edge
between $x\in V^-$ and a clause $D$ if at least one of the lifted
literals  $\bar x_1, \bar x_2$ appears in $D$. For a set $\mathbb D'\subseteq \mathbb D$
we denote by $N(\mathbb D')$ the set of its neighbors. 
We claim that there is a set of 
clauses $\mathbb D'\subseteq \mathbb D$ 
such that $|N(\mathbb D')|<|\mathbb D'|$. If this is not true,
then for all subsets $\mathbb D'\subseteq \mathbb D$, 
$|N(\mathbb D')|\geq |\mathbb D'|$, and 
by Hall's Theorem, there would be a set of edges $M$ in the incidence graph matching
$\mathbb D$ into $V^-$.  Using the set of edges in $M$ it is possible
to obtain a partial assignment $\alpha$ over $V^-[\vee]$ satisfying
all the clauses in $\mathbb D$ and with the property that for each pair of lifted variables $\{x_1,x_2\}$ for $x\in V^-$, at most one variable in each pair 
is assigned by $\alpha$.
We show how this 
assignment $\alpha$ can be extended to a total assignment $\alpha^*$ of $V[\vee]$  in such a
 way, that for some clause $C\in\NMproj(D)$, the formula $C[\vee]$ is falsified, which would be a contradiction.
 If there is a clause $C$ containing only negated literals, 
 this is done assigning 
 for each 
 literal  $\bar x$ in $C$  the twin variable $x_b$ of the variable $x_a$ in the
 domain of $\alpha$ with the value 1. Suppose for example that $C=(u\vee\bar x\vee \bar y\vee\bar z)$. Then the formula $C[\vee]$ is the conjunction of all the clauses
$(u_1\vee u_2\vee \bar x_a\vee\bar y_b \vee \bar z_c)$ with $a,b,c\in\{1,2\}$.
For such an assignment $\alpha^*$, there is a combination $a,b,c$ for which  $\alpha^*(u_1\vee u_2\vee \bar x\vee _a\vee \bar y_b \vee\bar z_c)=0$
and therefore $\mathbb D'\not\models C[\vee]$
contradicting the fact that $C\in \mproj(\mathbb D')$. 

In case all the clauses $\NMproj(\mathbb D)$ are axioms with some positive literal, let us consider 
one of such clauses $C$ with positive literal $x$. If $C[\vee]$ is satisfied by $\alpha$, then it must be because
$\alpha(x_1)=1$ or $\alpha(x_2)=1$. But since $x\in V^-$, there must be some other clause $C'\in \NMproj (\mathbb D)$ in which $x$ appears negated. $C'$ might have some variable $y$ such that $y_1\vee y_2$ is satisfied  by $\alpha$ but since $\bar x\vee y$ is part of an axiom in 
$\Peb_G$, $x$ is a predecessor of $y$ in $G$. Continuing with the argument, at some point we must reach an axiom $A\in \NMproj(\mathbb D)$ that is not satisfied by $\alpha$, and $\alpha $ can be extended to falsify $A$. 

So, there must be a set of clauses $\mathbb D'\subseteq \mathbb D$ with
$|N(\mathbb D')|<|\mathbb D'|$. If $\mathbb D'=\mathbb D$ then the result is clearly true. We show that no other case is possible. Let us suppose that
there is a set $\mathbb D'\subsetneq \mathbb D$ with 
$|N(\mathbb D')|<|\mathbb D'|$ and let $\mathbb D_1$ be a maximal set with this property. We define 
$\mathbb D_2=\mathbb D\setminus \mathbb D_1$, $V_1=N(\mathbb D_1)$, and 
$V_2 = V^-\setminus V_1$. By the hypothesis $V_2\not=\emptyset$. 
By the maximality of $\mathbb D_1$, for every
subset $\mathbb D''\subseteq\mathbb  D_2$ the number of neighbors of 
$\mathbb D''$ in $V_2$ must be greater  than $|\mathbb D''|$, 
since otherwise $\mathbb D_1$ would not be maximal having less neighbors than its cardinality.
Again, using Hall's Theorem, there must be a set of edges $M_2$,
matching $\mathbb D_2$ into $V_2$.
$M_2$ induces a partial  assignment $\alpha_2$ that satisfies all the clauses in $\mathbb D_2$ and
only assigns at most one of the twin variables  $x_1, x_2$ for each $x\in V_2$.
Consider now the set of clauses $S=\{C\in \NMproj(\mathbb D)$ with $\vars(C)\cap V_2\not=\emptyset\}.$
$S\not=\emptyset$ since $V_2\neq \emptyset$. We show that there is at least one clause $C\in S$ that is not satisfied by $\alpha_2$.
If there is a clause in $S$ with only negated literals for the variables in $V_2$, this is clearly the case. Otherwise, an argument as in the first part of the proof shows that there must be such a clause  $C$ that is not satisfied by $\alpha_2$.

By definition, there is a set of clauses $\mathbb{ \hat{D}}\subseteq\mathbb D$ with
$C\in\mproj(\mathbb {\hat{D}})$. We show that there is an assignment $\alpha$
that satisfies $\mathbb {\hat{D}}$ but falsifies $C[\vee]$, which is a contradiction. For this, let $\mathbb{\hat{D}}_1=\mathbb {\hat{D}}\cap \mathbb D_1$
and $\mathbb {\hat{D}}_2=\mathbb {\hat{D}}\cap \mathbb {{D}}_2$. Let also $C=C_1\cup C_2$
with $\vars(C_1)\subseteq V_1$ and $\vars(C_2)\subseteq V_2.$

By the definition of minimal projected clause we know that 
$\mathbb {\hat{D}}\models C[\vee]$ but $\mathbb {\hat{D}}\not\models C_1[\vee]$.
Since all the variables in $C_1$ are in $V_1$ we can infer that
$\mathbb{\hat{D}}_1\not\models C_1[\vee]$ and therefore there is partial assignment
$\alpha_1$ with domain $V_1[\vee]$ satisfying $\mathbb{\hat{D}}_1$ and falsifying $C_1[\vee]$.
The clauses in $\mathbb{\hat{D}}_2$ can be satisfied by $\alpha_2$ and therefore by $\alpha=\alpha_1\cup\alpha_2$.  
$\alpha$ satisfies all the clauses in 
$\mathbb{\hat{D}}$. Once more, 
$\alpha$ can be extended to  the twin variables of 
those assigned in $V_2$, falsifying $C_2[\vee].$
This way we obtain an assignment 
 that satisfies $\mathbb{\hat{D}}$ but falsifies $C[\vee]$, and
this is a contradiction since we supposed $C\in \mproj(\mathbb{\hat{D}})$.    
\end{proof}

We study now the clause space needed in an  N-Resolution refutation  of $\Peb_G[\vee]$. Unlike the negative refutations of $\Peb_G$, in which all the only clauses that contain positive literals are axioms, in the lifted formula the resolvent of an axiom and a negative clause 
contains one positive literal. This complicates the structure of the refutation. The next lemma states that at the price of one additional memory unit, a configurational refutation of $\Peb_G[\vee]$ can be transformed into another one in which the clauses with positive literals only stay in memory at most three configuration steps until they are removed again, and in fact, when an axiom is downloaded in memory, all the clauses already in memory contain only negative literals.

\begin{lemma}\label{lem:normal-form}
    Let $\pi$ be a configurational N-Resolution refutation of $\Peb_G[\vee]$ with $\cspace(\pi)=s.$ Then there is a an N-Resolution refutation 
    $\pi'$ with $\cspace(\pi')\leq s+1$ and satisfying that for every axiom $A\in\Peb_G[\vee]$ with a pair of positive literals $a_1, a_2$, if $A$ is downloaded in $\pi'$ at step $i$ in the refutation,
    then the clauses in the memory configuration at step $i-1$ do not contain any positive literal.
\end{lemma}

\begin{proof}
We describe how a configurational N-Resolution refutation of $\Peb_G[\vee]$ with $\cspace(\pi)=s$ can be transformed into an N-Resolution refutation 
    $\pi'$ with $\cspace(\pi')\leq s+1$ and satisfying that for every axiom $A\in\Peb_G[\vee]$ with a pair of positive literals $a_1, a_2$, if $A$ is downloaded in $\pi'$ at step $i$ in the refutation, then the following actions take place in the next configurations:
    \begin{itemize}
        \item Step $i+1$ is an inference step and $A$ is resolved with a negative clause $C=C'\vee \bar a_j\in \pi'$ for a $j\in\{1,2\}$, and the resolvent $R_1,$ is obtained.
        \item Step $i+2$, $A$ is deleted from memory,
        \item  Step $i+3$ is an inference step and $R_1$ is resolved with a negative clause in $\pi'$ (that contains the conjugated literal of the unique positive literal in $R_1$), obtaining $R_2$, and
\item Step $i+4$, $R_1$ is deleted from memory.
    \end{itemize}
        
   The condition that  the axiom $A$ stays in memory only for two steps is easy to achieve. Since we  want to minimize memory space, we do not need to download an axiom $A$ in memory
until right before it is needed in an inference step, and it can be deleted right in the next step. If it is needed again at some other step, it can be downloaded again, and this does not require any extra space.
So, in a first step, the refutation $\pi$ can be transformed into an intermediate one $\hat \pi$, with the same space requirements, and  with the property that when an axiom $A$
with two positive literals is downloaded at step $i$, then in step $i+1$
$A$ is resolved with some clause in memory, and in step $i+2$ $A$ is deleted from memory.

We now transform $\hat \pi$ into the goal refutation $\pi'$.
Let $R_1$ be a clause in memory with exactly one positive literal.
$R_1$ is the resolvent of an axiom $A$ and a negative clause $C$.
The idea is to keep $C$ in memory instead of $R_1$. For this we substitute every occurrence of $R_1$ in $\hat \pi$ by a copy of $C$ that we call $C_{R_1}$.
If $C$  and $C_{R_1}$ are the same clause, and in case both are in memory, they just count as one clause,  but it is convenient to think of $C_{R_1}$ as a place holder for $R_1$. 
Each time $R_1$ is needed for an inference (this can be more than once), $R_1$ can be inferred again by downloading $A$ in the proof, and resolving $A$ with $C_{R_1}$. This requires one extra unit of memory space for $A$, but once $R_1$ is derived, $A$ can me removed. In the next step
  $R_1$ is used for an inference  are deleted from memory (keeping $C_{R_1}$).
This implies that the extra clause in memory cannot interfere with other uses of $R_1$ or other clauses with positive literals in the refutation $\pi'$.
Once $R_1$ is removed from a configuration in $\hat\pi$, the clause $C_{R_1}$ is also  removed from   the configuration in $\pi'$ (keeping $C$ if this clause is still in memory).
\end{proof}

Observe that a consequence of the lemma is that in the refutation $\pi'$ in the first two steps, the only two negative axioms, $\bar z_1$ and 
$\bar z_2$ in $\Peb_G[\vee]$ corresponding to the sink vertex $z$ in $G$ are downloaded.    
Also, if for any $i$, step $i$ is a download step, then the memory configuration in $\mathbb M_{i-1}$ contains  only negative clauses.

The next result, together with Theorem~\ref{lem:cs>nproj}, implies that from a negative refutation $\pi$ of $\Peb_G[\vee]$ it is possible to extract a negative refutation
of $Peb_G$ with negative width bounded by the clause space of $\pi$.



\begin{lemma}\label{lem:negative-res}
    Let $\pi=\mathbb M_1,\dots, \mathbb M_t$ be a configurational negative Resolution refutation of $\Peb_G[\vee]$ satisfying the properties of
    Lemma~\ref{lem:normal-form}.
    For $i\in [t-1]$,  and for every clauses $C\in \Mproj(\mathbb M_{i+1})$ it is possible to derive $C$ by negative Resolution, 
    from $\NMproj(\mathbb M_{i})$ and axioms in $\Peb_G$ 
    and using only negated literals in $\NMproj(\mathbb M_i)\cup \NMproj(\mathbb M_{i+1})$.
\end{lemma}

\begin{proof}
By induction on $i$. Let $z$ be the sink in $G$.
As observed before,  $\NMproj(\mathbb M_1) =\emptyset$ and $\NMproj(\mathbb M_2) =\{\bar z\}$ which is an axiom in $\Peb_G$. 
For $i>2$, we only need to consider the case in which $i$ is an axiom download step. In an inference step, $\NMproj(\mathbb M_i)=\NMproj(\mathbb M_{i-1})$, and is a deletion step, no new  clauses can appear in $\NMproj(\mathbb M_i)$. If a new axiom $A\in \Peb_G[\vee]$ appears in the refutation it can correspond  to the sink vertex, to an internal vertex, or to a source vertex in $G$.

If $A$ corresponds to a sink vertex, for example $A=\bar z_1$ (this can happen if the axiom was deleted after the first two steps), then, since in $\mathbb M_{i-1}$ there are only clauses with negative literals, the only new clause that can appear in $\NMproj(\mathbb M_i)$ is $\bar z$, which is an axiom.

If $A$ corresponds to an internal  vertex, for example, $A=(\bar a_1 \vee\bar b_2\vee c_1\vee c_2)$, this clause is one of the clauses in the formula ${\cal A}[\vee]$
for the axiom ${\cal A}= (\bar a\vee \bar b\vee  c)$ in $\Peb_G$. 
Let us suppose that there is a clause $C\in \NMproj(\mathbb M_{i})\setminus \Mproj(\mathbb M_{i-1})$ (otherwise, nothing needs to be proven). If $C$ is an axiom in $\Peb_G$, then it can clearly be obtained by N-Resolution. In case $C$ is a negative clause
then,
$\mathbb M_{i-1}\not\models C[\vee]$. This is because the other possible reason 
for $C$ not being in $\NMproj(\mathbb M_i)$ would be that 
there is a subclause $C'\subsetneq C$ with 
$C'\in \proj(\mathbb M_{i-1})$, but  then also $C'\in \proj(\mathbb M_{i})$
implying $C\not\in \NMproj(\mathbb M_{i})$.

We claim that for any such clause 
$C$
there is 
a  subclause $C'\subseteq  C$ such that
$C'\vee \bar c\in \NMproj(\mathbb M_{i-1})$ and moreover the negated literals $\bar a$ and $\bar b$ belong to $C$.

The result follows from the claim since we can resolve $C'\vee \bar c$ with the axiom $(\bar a\vee \bar b\vee c)$ (which is a negative Resolution step
since $C$ is a negative clause) and obtain the clause $C''\subseteq C$. Observe that all the variables used in the Resolution step appear negated in $\NMproj(\mathbb M_{i-1})\cup \NMproj(\mathbb M_{i})$.

We  prove the claim. 
As stated above $\mathbb M_{i-1}\not\models C[\vee]$. Let $\alpha$ be an assignment satisfying $\mathbb M_{i-1}$ 
and falsifying 
$C[\vee]$, since  by assumption $C\in\Mproj(\mathbb M_i\cup \{A\})$
this implies $\alpha(A)=0$ and therefore $\alpha(c_1)=\alpha(c_2)=0$, which implies 
$\mathbb M_{i-1}\models C\vee \bar c[\vee]$  and therefore for some subclause $C'\subseteq C$, $C'\vee \bar c\in \NMproj(\mathbb M_{i-1})$.

Also the literal $\bar a$ must be in $C$. Arguing in the same way as for $\bar c$, $\mathbb M_{i-1}\models C\vee a[\vee]$. 
But by the conditions of the refutation (Lemma \ref{lem:normal-form}), the clauses in $\mathbb M_{i-1}$ 
contain only negative literals and can only (minimally) project negative clauses (or the tautology 1 which is always a logical consequence). The only way to avoid a contradiction is that $\bar a\in C$ which means
$C\vee a$ is tautological, and therefore implied by $\mathbb M_{i-1}.$
The same argument applies to literal $\bar b$.
This proves the claim.

The only case missing is when the downloaded axiom $A$ corresponds to a source vertex, $A=(a_1\vee a_2)$. This is exactly as the first part of
previous case (when dealing with $c_1\vee c_2$), and in this case the resolvent of some negative clause in $\NMproj(\mathbb M_{i-1})$ with the axiom
$a\in \Peb_G$ is a subclause of $C$.
\end{proof}

Lemmas~\ref{lem:cs>nproj} and \ref{lem:negative-res} imply:

\begin{theorem}\label{thm:translated-refutation}
    Let $G$ be a sDAG  and $\pi$ be  an N-Resolution refutation of $\Peb_G[\vee]$ with $\clspaceabbr(\pi)=s$.  There is a negative Resolution refutation of $\Peb_G$ with negative width at most $s$.  
\end{theorem}

The last step in the proof of the lower bound  shows that for a sDAG $G$, from a N-Resolution refutation for $\Peb(G)$ with negative width $w$ it is possible to extract a strategy for the black pebble game on $G$ using at most $w+1$ pebbles. For this, it is useful to consider the dual white pebble game. In this version of the game, a pebble can be placed on any vertex of $G$ at any time, but it can only be removed from a vertex if all its predecessors have a pebble on them. The game starts with a pebble on the sink vertex and ends when all the pebbles have been removed from $G$. As mentioned in the preliminaries, reversing the order of a black pebbling strategy on a graph $G$, one obtains a strategy with the same number of pebbles in the white pebble game  and therefore, for any graph $G$, the black and the white pebbling numbers coincide.

\begin{lemma}\label{lem:white-pebbling}
For any DAG $G=(V,E)$ with a single sink $z$, the negative-width in an N-Resolution refutation for $\Peb_G$ is at least $\Black(G)-1$.
\end{lemma}
\begin{proof}
    We recall that the formula $\Peb(G)$ is a Horn formula and therefore all its clauses have at most one positive literal. This implies that a negative refutation must start with the only negative clause $(\bar x_z)$ and this clause is resolved with another axiom producing a new negative clause. In each Resolution step $i$, the negative clause $C_i$ obtained in the previous step is resolved with an axiom, producing a new negative clause $C_{i+1}$. This is the only way to reach the empty clause. This kind of strategy can be transformed into a strategy for playing the white pebble game on $G$. Starting with a white pebble on the sink vertex $z$, at each Resolution step $i$ there are white pebbles
    exactly on the vertices corresponding to the literals in the negative clause $C_i$. It is possible to go from $C_i$ to $C_{i+1}$ following correct white pebbling steps.  This is done first placing a pebble on the vertices corresponding to variables in the axiom resolved with $C_i$ that do not have a pebble and then removing the
 pebble corresponding to the resolved variable. This is possible since the variables for the predecessor of the
    resolved vertex appear in the axiom resolved with $C_i$, and therefore, their corresponding vertices have a pebble on them. The refutation ends resolving a source axiom, and therefore the white pebble on its vertex can be removed. At each step, we need as many pebbles as literals in the clause $C_i$ plus one more for the vertex on the  resolved variable. Therefore, the pebbling number of $G$ is bounded by 
    the negation-width of $\Peb_G$ in N-Resolution, plus one.    
\end{proof}

The lower bound in the theorem follows directly from the two previous results.

\begin{theorem}Let $G$ be a sDAG. 
$\Black(G)\leq \Nclspaceref{\Peb_G[\vee]}$.
\end{theorem}

For the case of the $\Peb_G[\oplus]$ formulas, the same lower bound holds, with a  simple proof.

\begin{theorem}
Let $G$ be a sDAG. 
$\Black(G)\leq \Nclspaceref{\Peb_G[\oplus]}$.
\end{theorem}

\begin{proof}
Consider the partial assignment $\alpha$ that for each pair of variables  $v_1,v_2$ in ${\Peb_G[\oplus]}$
gives value 1 to $v_1$ and leaves $v_2$ unassigned. ${\Peb_G[\oplus]}\restriction_\alpha$ is a renaming of the formula $\Peb_G$ 
in which each literal has been substituted by its negation. In Theorem~\ref{thm:tree-pcs} this formulas is called $\widehat{\Peb_G}$ and it is shown show that $\Black(G)\leq \Nclspaceref{\widehat{\Peb_G}}$. Since
for any unsatisfiable formula $F$ and any partial assignment $\alpha$, a refutation of $F$ cannot use less space than a refutation of 
$F\restriction_\alpha$, the result follows.

\end{proof}

\section{Comparing space in tree-like and negative Resolution}\label{sec:comparing}

In this section, we study how different the space measure can be in the two restrictions of Resolution.
The best known separation in space between general and tree-like Resolution is by a $\log n$ factor \cite{JMNZ12RelatingProofCplx, ToranW21}, there are formulas $F$  for which
$\treeclspaceref[]{F}=\Theta(\log n)$ and  $\clspaceref[]{F}=O(1),$ and it is open whether this can be improved.  
The same separation between tree-like and negative Resolution space can be achieved.
\begin{theorem}
    There is a family of  formulas $\{F_n\}$ with $O(n)$ variables for which\\ 
    $\treeclspaceref[]{F_n}=\Theta(log n)$ and  $\Nclspaceref[]{F_n}=O(1).$  
\end{theorem}

\begin{proof}
    Let $G_n$ be the  road graph length $n$ and width 2 and let $F_n=\Peb_{G_n}[\vee]$. Starting from the sink axioms, and moving towards the sources of $G_n$ it is possible to refute the formula with constant space in general Resolution. For this at each level with two vertices $v, v'$ in $G_n$ one has to  keep in memory the four clauses $(\bar v_i\vee \bar v'_j)$, $i,j\in\{1,2\}$. 
On the other hand, by Lemma \ref{lem:road-graph} and Theorem \ref{thm:tree-cs-peb},  $\treeclspaceref[]{F}=\Theta(\log n)$.
\end{proof}

We can achieve a much larger separation in the other direction. This is interesting since it is known that 
N-Resolution polynomialy simulates tree-like Resolution in size (see e.g. \cite{BP07Complexity}). It is easier to  first analyze 
the case of positive Resolution, then the result can be adapted to negative Resolution

\begin{theorem}\label{thm:tree-pcs}
There are unsatisfiable formula families  $\{F_n\}$, $\{H_n\}$  with $O(n)$ variables   for which 
\begin{enumerate}
\item $\treeclspaceref[]{F_n}=O(1)$,   
but $\Pclspaceref[]{F_n}=\Omega(\frac{n}{\log n})$, and  
\item
 $\treeclspaceref[]{H_n}=O(1)$ but $\Nclspaceref[]{H_n}=\Omega(\frac{n}{\log n})$. 
\end{enumerate}
\end{theorem}
\begin{proof}
For the first part, we consider the pebbling formula $\Peb_G$ for a sDAG $G$ with in-degree at most 2. It is well known that resolving the clauses following a topological ordering of the underlying graph $G$,  $\Peb_G$ has a tree-like refutation with clause space O(1). This is in fact a negative Resolution refutation.
We will show that from a positive refutation $\pi$ of $\Peb_G$ with clause space $k$ one can extract a black pebbling strategy $S_\pi$ for $G$ with 
$\Black(S_\pi)\leq k.$ The result follows  from the fact that there are known families of expander graphs  $G_n$ with $n$ vertices and high pebbling number, $\Black(G_n)=\Omega(\frac{n}{\log n})$ \cite{PTC76SpaceBounds}.

We observe that any positive Resolution of a Horn formula must be a unit Resolution. This is because every derived clause in the refutation of a Horn formula must be a Horn clause, and therefore the only positive clauses in the refutation are unit clauses.

Depending on the step taken in the configurational refutation, we perform some pebbling actions in the following way:

\begin{enumerate}
\item{}If a source axiom $a$ is downloaded, then a pebble is placed on vertex $a$.

\item{}If a unit clause $c$ is deleted and $c$ has been resolved with some other clause $A$ containing $\bar c$ and some positive literal $d$, then the pebble on $a$ is not deleted until the clause $d$ has been derived.

\item{}If a Resolution step is performed obtaining a unit clause with a positive literal $d$, then a pebble is placed on $d$. 
\end{enumerate}

We show that this is a correct pebbling strategy, and that the number of pebbles used at any moment $t$ in the strategy is 
bounded by the number of clauses in $\mathbb M_t.$

For the correctness of the strategy, observe that the only moments in which new pebbles are set are when a source axiom $a$ is downloaded or a
positive literal $d$ is derived. In the first situation, it is clear that one can place a pebble on $a$. In the second case,
the literal $d$ comes from an axiom of the kind $(\bar a\vee \bar b\vee d)$ and two Resolution steps with the clauses $a$ and $b$ have been performed to derive $d$.
 Due to Rule 2, the pebbles on $a$ and $b$ are still set, and therefore it is possible to place a pebble on $d$.
Also, at some point, a pebble has to be placed on the sink $z$ since the positive literal $z$ has to be derived in $\pi$ because
this is the only clause that can be resolved with the axiom $\bar z$ and without this axiom the formula is satisfiable.

We show now that the number of pebbles at any moment $t$ is bounded by $|\mathbb M_t|$. Observe that if at time $t$ 
there is a pebble on vertex $u$, then either the clause $u\in \mathbb M_t$ or there is a clause $(\bar v\vee w)$ in $\mathbb M_t$ obtained from the Resolution of $u$ and some other clause in a previous configuration. In this case we can associate  $u$ with this clause since no other unit clause can be resolved to obtain the resolvent $(\bar v\vee w)$.
 
For the part about N-Resolution,    let us define for a formula $F$, the modified formula  $\widehat F$ obtained by substituting each literal in $F$ by its negation, that is, each positive literal in $F$ is negated in $\widehat F$ and each negated literal is transformed into its positive version. A positive Resolution refutation $\pi$ of $F$ corresponds then to a negative refutation of 
$\widehat F$. Therefore, by the previous result, for a DAG $G$, $\Nclspaceref{\widehat {\Peb_G}}=\Omega(\Black(G))$. 

Observe that the families $F_n$ and $H_n$ provide also a way to separate the space requirements in N- and P-Resolution.
\end{proof}

\section{Conclusions and open problems}
We have introduced the blue pebbling price,   a new cost  measure for pebbling games, intermediate between the well known 
black and reversible pebbling measures. It would be interesting to find graph families for which all three measures are 
pairwise asymptotically different.
We have shown that the blue price is a useful tool for the analysis of 
clause space in Resolution, exactly characterizing this measure for general unsatisfiable formulas in tree-like Resolution  improving existing results. We have also considered the space in negative Resolution  obtaining tight upper and lower bounds 
related to pebbling for 
different classes of lifted pebbling formulas. The bounds, which  not  completely match, can be seen in Table 1. It is open whether they can be improved in terms of pebbling game measures. 
Finally we have obtained close to optimal space separations between negative  and tree-like Resolution.
We do not know of any other case of two proof systems $A$ and $B$ in which $A$ can simulate $B$ in size,
but requires much higher clause space than $B$ for refuting certain formulas.

\bibliography{references.bib}
\bibliographystyle{alpha}
\end{document}